%% file: UniversalDecoding_ISIT.tex
\documentclass[12t,conference,onecolumn]{IEEEtran}  

\input{common_defs.tex}

\newif\ifISITpaperFull
\ISITpaperFulltrue 

\begin{document}

\title{A Universal Decoder Relative to a Given Family of Metrics}

\author{\authorblockN{Nir Elkayam and Meir Feder}
\authorblockA{Department of Electrical Engineering - Systems\\
Tel-Aviv University, Israel \\
Email: nirelkayam@post.tau.ac.il, meir@eng.tau.ac.il}}

\maketitle

\let\thefootnote\relax\footnotetext{This work was partially supported by the Israeli Science Foundation (ISF), grant no. 634/09.}

\subsection*{\centering Abstract}
\textit{
Consider the following framework of universal decoding suggested in \cite{merhav:universal}. Given a family of decoding metrics and random coding distribution (prior), a single, universal, decoder is optimal if for any possible channel the average error probability when using this decoder is better than the error probability attained by the best decoder in the family up to a subexponential multiplicative factor.
We describe a general universal decoder in this framework. The penalty for using this universal decoder is computed. The universal metric is constructed as follows. For each metric, a {\em canonical} metric is defined and conditions for the given prior to be {\em normal} are given. A sub-exponential set of canonical metrics of normal prior can be merged to a single universal optimal metric. We provide an example where this decoder is optimal while the decoder of \cite{merhav:universal} is not.
}

\section{Introduction}\label{sec:intro}

\nocite{ElkayamEE2014}
Metric based decoders are usually used to decode codes in digital communication. Typically, the metric tries to capture the most likely codeword that was transmitted. When the channel over which the codeword was transmitted is known at the receiver, the \emph{Maximum Likelihood} (ML) decoder is used for optimal performance (in terms of block error rate). When the channel is not known at the receiver, other solutions are required.

Usually another assumption is made, for example, the channel $W_{\theta}(\by|\bx)$ belongs to some family of channels indexed by $\theta \in \Theta$. The family of channels can have some measure defined over it (Bayesian approach), \eg, fading channel, or the channel is one of several possible channels (deterministic approach), \eg, discrete memoryless channels (DMCs). The performance in the first case can be measured over the channel realization, and in the latter per specific channel.

A universal decoding is \emph{optimal} with respect to a given family of channels if the performance when using the universal decoder is not worse (in the error exponent sense) than when the optimal decoder is used; specifically, the error exponent is not smaller than that of the optimal decoder.

In \cite{goppa1975nonprobabilistic}, Goppa offered the \emph{maximum mutual information} (MMI) decoder, which decides in favor of the codeword having the maximum empirical mutual information with the channel output sequence. Other universal decoders have been suggested over the years for ISI channels \cite{DBLP:journals/corr/HuleihelM14, lapidoth2000gaussian, farkas2008blind},Interference channels \cite{merhav1993universal}, finite state channels \cite{ziv1985universal}, \cite{lapidoth1998universality}, Individual channels \cite{DBLP:journals/tit/LomnitzF13}, and \cite{DBLP:journals/tit/LomnitzF13a}. In \cite{feder1998universal}, Feder and Lapidoth provide a fairly general construction of a universal decoder for a given class of channels. Their construction has two key points. The first is the use of uniform random coding distribution (prior), which means that all codewords have the same probability; this allows the use of the \emph{merged decoder}. The second is the definition of a separable family, which means that there is a countable set of channels that can be used as representative, in the sense that if the performance of the universal decoder is good over these channels, then it would be good over the whole family of channels.

In \cite{merhav:universal} Merhav proposed another framework of universal decoding, namely, universality relative to a given set of decoding metric over arbitrary channels, where optimality means that the performance of the universal decoder is not worse (in the error exponent sense) than the performance when any other metric from the family is used. The decoder proposed by Merhav in \cite{merhav:universal} uses an equivalent of ``types'' that is a set of input and output sequences of the channel for which the metrics of the family provides the same value. The universal decoder that is based on ties in the metric, thus, does not seem to be the most general.

In this paper we generalize the construction in \cite{merhav:universal}, which in some sense also generalizes Feder and Lapidoth`s universal decoder, \cite{feder1998universal}. The quantity $\pem$, which is the error probability when using the metric $m$, given that $\bx$ was transmitted and $\by$ was received, is defined. This quantity plays an important role in the computation of the average error probability and can be used to compare the performance between different metrics. Based on $\pem$ we will define, for a metric $m$, a ``canonical metric'' $\overline{m}\BRA{\bx,\by}=-\log \pem$, which is equivalent to $m$. The canonical metrics (or the associated error probabilities $\pem$) serve, in turn, to define a universal decoder relative to a set of metrics, called the \textbf{Generalized Minimum Error Test}. We derive the penalty of using this universal metric, and show when this metric is optimal (in the error exponent sense).

The definition of the canonical metric is related to the definition in \cite{lomnitz2012communication} of \emph{empirical rate function}. This rate function, which is a function of the transmitted word $\bx$ and the received word $\by$, was defined in an attempt to remove the probabilistic assumption in communication and get a rate function that fits the individual sequences of input and output. Specifically, after we introduce the notion of normal prior, we show that for every metric $m$ over a normal prior there exists a ``tight'' empirical rate function, \ie, the canonical metric, and $\overline{m}\BRA{\bx,\by} \approx \log\BRA{ \frac{P(\bx|\by)}{Q(\bx)}}$. Here $\approx$ means equal up to a sub-exponential factor and $P(\bx|\by)$ is a probability distribution over $\cX$. A set of countable ``tight'' rate functions can be merged into a single universal decoding metric that is optimal with respect to all the decoding metrics in the set. This is the generalization of the merged decoder of \cite{feder1998universal}.

\section{Notation}

This paper uses bold lower case letters (\eg, $\bx$) to denote a particular value of the corresponding random variable denoted in capital letters (\eg, $\bX$). Calligraphic fonts (\eg, \calX) represent a set.
Throughout this paper $\log$ will be defined as base 2 unless otherwise indicated. $\PR{A}$ will denote the probability of the event $A$. We will use the $o$ notation. Specifically, $o(n)$ will be used to indicate a sequence of the form: $n \cdot \lambda_n$ where $\lambda_n \toinf 0$ and $o(1)$ is a sequence that goes to 0. A sequence of positive numbers $\chi_n$ is said to be \textbf{subexponential} if $\limtoinf \frac{1}{n}\log(\chi_n) = 0$. With the $o$ notation this is just $\chi_n=2^{o(n)}$


\section{Preliminaries}

A rate-$R$ blocklength-$n$, \textbf{random code} with prior $Q_n(\bx)$ is:
\begin{align}
\cC = \BRAs{\bx(1), \bx(2), ..., \bx(2^{nR}) } \subset \cX^n,
\end{align}
where each codeword $\bx(i)$ is drawn independently according to the distribution $Q_n(\bx)$. The \textbf{encoder} maps the input message $k$, drawn uniformly over the set $\BRAs{1, ..., 2^{nR}}$, to the appropriate codeword, \ie, $\bx(k)$. A \textbf{decoding metric} is a function $m: \cX^n \times \cY^n \to \mathbb{R}$. The \textbf{decoder} associated with the decoding metric $m$ is the function $\cD_m: \cY^n \to \BRAs{1, ..., 2^{nR}}$ defined by:

\begin{align}
\cD_m(\by) = \hat{i} = argmax_{1 \leq i \leq 2^{nR}} \BRA{m(\bx(i),\by)}.
\end{align}

\begin{definition}
The \textbf{average error probability} associated with the decoder $\cD_m$ over the channel $W_n(\by|\bx)$ is denoted by $\bar{P}_{e,m,W}\BRA{R}$. The error probability captures the randomness in the codebook selection, the transmitted message, and the channel output.
\end{definition}

\begin{remark}
We will generally refer to a metric as a decoder where it would be understood that we refer to the decoder associated with the given metric.
\end{remark}

\begin{remark}
When the decoding metric $m(\bx,\by) = W_n(\by|\bx)$, the error probability is minimized and the decoder is called the Maximum Likelihood decoder. The \textbf{mismatch} case is when the metric is not matched to the channel and may result in performance degradation.
\end{remark}

\begin{definition}
For a family of decoding metrics:
\begin{align}\label{def:met_familty}
\cM_n = \BRAs{m_{\theta}(\bx,\by) : \theta \in \Theta_n, \bx \in \cX^n, \by \in \cY^n}.
\end{align}
Denote:
\begin{align}\label{def:met_error_probability}
\bar{P}_{e,\cM_n,W_n}\BRA{R} = \min_{\theta \in \Theta_n}\bar{P}_{e,m_{\theta},W_n}\BRA{R},
\end{align}
the minimum error probability over the channel $W_n(\by|\bx)$ when the best metric matched for this channel is used.
\end{definition}

\begin{definition}
A sequence of decoding metrics $u_n$, independent of $\theta$, and their associated decoders is called \textbf{universal}. The sequence of universal decoding metrics $u_n$ (and their decoders) is \textbf{optimal} (in the error exponent sense) with respect to the family $\cM_n$ if:
$$\limtoinf \frac{1}{n}\log\BRA{\frac{\bar{P}_{e,u_n,W(\by|\bx)}\BRA{R}}{\bar{P}_{e,\cM_n,W(\by|\bx)}\BRA{R}}} \leq 0$$
for every channel $W(\by|\bx)$.
\end{definition}

Consider a metric $\metric{m}$. The pairwise error given that $\bx$ was transmitted and $\by$ was received is given by:
\begin{align} \label{def:pair_err}
\pem = Q_n(m(\bX,\by) \geq m(\bx,\by)).
\end{align}

The probability given by \eqref{def:pair_err} is for a single (random) codeword $\bX$, to beat the transmitted word $\bx$ given that the received word was $\by$. Notice that in case of a tie, we are assuming that an error has occurred.


The term $\pem$ can be used to compute the average error probability for any rate $R$ as can be seen in the following lemma.
\begin{lemma} \label{lem:comm_total_error}
\begin{equation}\label{eq:union_bound_cliped}
 \frac{1}{2} \leq \frac{\bar{P}_{e,m,W}\BRA{R}} {\E{\min\BRA{1, 2^{nR} \cdot \pem}}} \leq 1,
\end{equation}
where the expectation is with respect to $Q(\bx)W(\by|\bx)$.
\end{lemma}

This lemma is given in \cite{merhav:universal}. The upper bound follows from the union bound. The lower bound follows from a lemma proved by Shulman in \cite[Lemma A.2]{ShulmanPHD}; namely, that the union clipped to unity is tight up to factor 2 (for pairwise independent events).


\begin{definition}
We say that the metric $u_n$ \textbf{dominates} $m_n$ if for all $\bx,\by$:
\begin{equation}\label{def:dominate}
\pe{u_n} \leq \pe{m_n} \cdot 2^{n \cdot \lambda_n}
\end{equation}
with with $\lambda_n \toinf 0$.
\end{definition}

The following lemma, which also appears in \cite{merhav:universal}, follows from lemma \eqref{lem:comm_total_error}:
\begin{lemma} \label{lem:preformence_compre}
\begin{align} \label{preformence_compre}
\bar{P}_{e,u_n,W}\BRA{R} \leq 2\cdot 2^{n \cdot \lambda_n} \cdot \bar{P}_{e,m_n,W}\BRA{R} \\
\bar{P}_{e,u_n,W}\BRA{R} \leq 2\cdot \bar{P}_{e,m_n,W}\BRA{R+\lambda_n}.
\end{align}
In particular, if metric $u_n$ dominates $m_n$ then the average error exponent when using the metric $u_n$ is not worse than using the metric $m_n$.
\end{lemma}

Obviously, an optimal universal decoding metric is a metric that dominates all the metrics in the family \eqref{def:met_familty}.


\section{ The Generalized minimum error test} \label{sec:dec_comp}

In this section we propose our universal decoder, the \textbf{Generalized Minimum Error Test (GMET)}. The decoder estimates the metric that minimizes the pairwise error probability.
\begin{definition}
For the given family of metrics \eqref{def:met_familty}, let:
\begin{equation}\label{def:U}
2^{-n\cdot U_n(\bx,\by)} \triangleq \min_{\theta \in \Theta_n} \pe{m_{\theta}}.
\end{equation}
$U_n(\bx,\by)$ is the GMET. Let:
\begin{equation}\label{def:redundancy}
K_n(\Theta_n) \triangleq \MAX{\by} \; \EQn{2^{n\cdot U_n(\bX,\by)}}.
\end{equation}
$K_n(\Theta_n)$ is termed the \textbf{redundancy} associated with the decoder.
\end{definition}

The performance degradation of the proposed universal decoder relative to any metric $m_{\theta}(\bx,\by)$ in the family can be measured through the redundancy as demonstrated in the following Theorem.

\begin{theorem} \label{theorem:universal_decoder}
For any $\theta \in \Theta_n$
\begin{equation}\label{eq:U_dominate}
\pe{U_n} \leq \pe{m_{\theta}} \cdot K_n(\Theta_n).
\end{equation}
In particular, if $K_n(\Theta_n)$ is sub-exponential then the universal decoder $U_n(\bx,\by)$ is optimal.
\end{theorem}

\ifISITpaperFull
\begin{proof}
\begin{align*}
  \pe{U_n} &=                    Q_n\BRA{U_n(\bX,\by) \geq U_n(\bx,\by)} \\
           &=                    Q_n\BRA{2^{n\cdot U_n(\bX,\by)} \geq 2^{n\cdot U_n(\bx,\by)}} \\
           &\overset{(a)}{\leq} \EQn{2^{n\cdot U_n(\bX,\by)}} \cdot 2^{-n\cdot U_n(\bx,\by)} \\
           &\overset{(b)}{\leq}  \pe{m_{\theta}} \cdot K_n(\Theta_n) \\
\end{align*}
where (a) follows from Markov's inequality and (b) from the definition of $K_n(\Theta_n)$ \eqref{def:redundancy} and $U_n(\bx,\by)$ \eqref{def:U}.
If $K_n(\Theta_n)$ is sub-exponential then $U_n(\bx,\by)$ dominates every metric $m_{\theta}$. By Lemma \ref{lem:preformence_compre} the universal decoder is optimal.
%
\end{proof}
\fi

Theorem \ref{theorem:universal_decoder} provides a quite general construction of a universal decoder. However, bounding $K_n(\Theta_n)$ might not be an easy task.

\subsection{Conditions for universal decoding}
Obviously, for a family of a \emph{single} metric there must be a universal decoder, \ie, the same single metric or an equivalent metric. By analyzing this simple case we get a core understanding of when the proposed GMET decoder is optimal.

\begin{definition}
For a given matric $m(\bx,\by)$ let:
\begin{equation}\label{def:canonic_metric}
\bar{m}(\bx,\by) \triangleq -\frac{1}{n}\log\BRA{\pem}.
\end{equation}
$\bar{m}(\bx,\by)$ is called \textbf{the canonical metric} associated with $m(\bx,\by)$.
\end{definition}
The canonical metric $\bar{m}$, which is equivalent to the original metric $m$ as it induces the same order of the candidate words, is the proposed universal decoder for the family of a single metric $m(\bx,\by)$. The next definition captures the cases where this universal metric is optimal in the sense of Theorem \ref{theorem:universal_decoder}.
\begin{definition} \label{def:normal} \mynewline
\begin{enumerate}
  \item The metric $m(\bx,\by)$ over the prior $Q_n(\bx)$ is \textbf{normal} if
  \begin{equation}\label{def:NormalMetric}
    \EQn{2^{n\cdot \bar{m}(\bX,\by)}} \leq 2^{n\cdot\lambda_n}
  \end{equation}
  uniformly over $\by$, where $\lambda_n \toinf 0$.
  \item The prior $Q_n(\bx)$ is normal if each metric over $Q_n(\bx)$ is normal (uniformly over all the metrics, \ie, the bound in \eqref{def:NormalMetric} is independent of the metric).
\end{enumerate}
\end{definition}
Normal metrics are exactly those metrics that might admit an optimal universal decoder.

The following lemma provides a sufficient condition for a given discrete prior $Q_n(\bx)$ to be normal. The condition is very mild and easy to check.
\begin{lemma}
Let $\chi_n = -\log \BRA{ \min_{\bx : Q_n(\bx) \neq 0}Q_n(\bx)}$. If $\chi_n$ is sub-exponential, then $Q_n(\bx)$ is normal.
\end{lemma}
\begin{proof}
The proof is given in the appendix (Proposition \ref{prop:normal_prior}) where we also introduce the concept of ``rate function'' and relate it to the canonical metric.
\end{proof}

\subsection{The Merged decoder}
If the family of metrics \eqref{def:met_familty} is finite, we can "merge" all the metrics and get a bound on $K_n(\Theta_n)$. If the family size grows sub-exponentially over the normal prior, then the family admits a universal decoder.

\begin{theorem}
Suppose $Q_n(\bx)$ is normal and $\Theta_n$ is finite, then:
\begin{enumerate}
  \item $K_n(\Theta_n) \leq |\Theta_n|\cdot 2^{o(n)}$
  \item If $|\Theta_n|$ grows sub-exponentially, then the family $\Theta_n$ admits an optimal universal decoder.
\end{enumerate}
\end{theorem}

\ifISITpaperFull

\begin{proof}
For each $\by$:
\begin{align*}
\EQn{2^{n\cdot U_n(\bX,\by)}} &= \EQn{\frac{1}{\min_{\theta \in \Theta_n } \pe{m_{\theta}} }} \\
&\overset{(a)}{\leq} \sum_{\theta \in \Theta_n } \EQn{\frac{1}{\pe{m_{\theta}}}} \\
&\overset{(b)}{\leq} |\Theta_n| \cdot 2^{o(n)}
\end{align*}
(a) follows by taking $\max$ over all $\by$ and (b) from Theorem \eqref{theorem:universal_decoder}.
\end{proof}
\fi

\subsection{Approximation of the universal decoder}
As noted also in \cite{merhav:universal}, the proof actually gives conditions for a universal decoding metric to be optimal. Suppose that there exists another universal metric $U_n'(\bx,\by)$ such that:
\begin{equation}\label{Approximate:App}
2^{-n\cdot U_n'(\bx,\by)} \leq 2^{-n\cdot U_n(\bx,\by)+o(n)}
\end{equation}
and
\begin{equation}\label{Approximate:Kraft}
\EQn{2^{n\cdot U_n'(\bx,\by)}} \leq 2^{o(n)}.
\end{equation}
Then it follows that:

\begin{align*}
  \pe{U_n'} &=                    Q_n\BRA{U_n'(\bX,\by) \geq U_n'(\bx,\by)} \\
            &\overset{(a)}{\leq} \EQn{2^{n\cdot U_n'(\bX,\by)}} \cdot 2^{-n\cdot U_n'(\bx,\by)} \\
            &\overset{(b)}{\leq} 2^{o(n)} \cdot 2^{-n\cdot U_n(\bx,\by)} \\
            &\overset{(c)}{\leq}  \pe{m_{\theta}} \cdot K_n(\Theta_n) \cdot 2^{o(n)}\\
\end{align*}
where (a) follows from Markov's inequality, (b) is the given conditions, and (c) is Theorem \ref{theorem:universal_decoder}. In particular, if $K_n(\Theta_n)$ is subexponential, then $U_n'(\bX,\by)$ is also optimal.
Now,
\begin{align*}
  \pe{m_{\theta}} &=    Q\BRA{m_{\theta}(\bX,\by) \geq m_{\theta}(\bx,\by)} \\
                  &\geq Q\BRA{m_{\theta}(\bX,\by)   =  m_{\theta}(\bx,\by)} \\
                  &\geq Q\BRA{\forall \tau\in\Theta_n, m_{\tau}(\bX,\by) =  m_{\tau}(\bx,\by)} \\
\end{align*}
So there are several general approximations that might be used instead of the original proposed decoder, \ie:
\begin{equation}\label{Decoder:App1}
2^{-n \cdot U_{n,1}(\bx,\by)} = \MIN{\theta \in \Theta_n}Q\BRA{m_{\theta}(\bX,\by)   =  m_{\theta}(\bx,\by)}.
\end{equation}
and
\begin{equation}\label{Decoder:App2}
2^{-n \cdot U_{n,2}(\bx,\by)} = Q\BRA{\forall \tau\in\Theta_n, m_{\tau}(\bX,\by)   =  m_{\tau}(\bx,\by)},
\end{equation}
where in the definition of $U_{n,2}(\bx,\by)$ we can avoid the minimization over $\theta$ as it is effectively done in the probability computation. For each of these decoders the condition \eqref{Approximate:Kraft} should be checked and universality(optimal) of the approximation implies universality of the GMET decoder, but not vice versa (\cf\ example \ref{Example:FSC}). Notice that $U_{n,2}(\bx,\by)$ is the decoder proposed by Merhav in \cite{merhav:universal}. These approximations might be useful when computation of the exact GMET is hard and these approximation are easier to calculate.

\ifISITpaperFull

\begin{remark}
The decoder $U_{n,2}(\bx,\by)$ has a nice property of being a ``asymptotically tight rate function'' (when it is optimal) in the following sense:
The use of the Markov's inequality gives the upper bound on $\pe{U_{n,2}}$:
\begin{align*}
  \pe{U_{n,2}} \leq 2^{-n\cdot U_{n,2}(\bx,\by)} \cdot 2^{o(n)}.
\end{align*}
However, for $U_{n,2}$ it is also true that $\pe{U_{n,2}} \geq 2^{-n\cdot U_{n,2}(\bx,\by)}$. To see this, recall from \cite[Eq. (4),(5)]{merhav:universal}:
\begin{equation*}
\mathcal{T}_n(\bx|\by) = \BRAs{ \bx' \in \cX^n : \forall \theta \in \Theta_n, m_{\theta}(\bx',\by) = m_{\theta}(\bx,\by)}
\end{equation*}
and
\begin{equation*}
2^{-n\cdot U_{n,2}(\bx,\by)} = Q_n\BRA{\mathcal{T}_n(\bx|\by)}.
\end{equation*}
Then:
\begin{align*}
  \pe{U_{n,2}}  &=     \sum_{\bx' : U_{n,2}(\bx',\by) \geq U_{n,2}(\bx,\by) } Q_n(\bx') \\
                &\geq  \sum_{\bx' : U_{n,2}(\bx',\by)   =  U_{n,2}(\bx,\by) } Q_n(\bx') \\
                &\overset{(a)}{\geq}  \sum_{\bx' \in \mathcal{T}_n(\bx|\by) } Q_n(\bx') \\
                &= Q_n\BRA{\mathcal{T}_n(\bx|\by)} \\
                &= 2^{-n \cdot U_{n,2}(\bx,\by)}
\end{align*}
where (a) follows because $U_{n,2}(\bx,\by) = U_{n,2}(\bx',\by)$ for $\bx' \in \mathcal{T}_n(\bx|\by)$.
Putting these together we have:
\begin{equation}\label{App:Tight}
-\frac{1}{n}\log\BRA{\pe{U_{n,2}}} = U_{n,2}(\bx,\by)+o(1)
\end{equation}
So we see that the metric $U_{n,2}$ is almost in the canonical form (up to a vanishing term), which means that we can use it to evaluate the performance of the decoder (instead of using $\pe{U_{n,2}}$). It is not apparent that $U_n$ in general has the ``asymptotically tight'' property of $U_{n,2}$. This means that $U_n(\bx,\by)$ as a rate function provides only the lower bound on performance, and it might be that the performance is better in the sense that $-\frac{1}{n}\log\BRA{\pe{U_n}}$ might be larger than $U_n(\bx,\by)$ by a non-vanishing term.
\end{remark}

\fi

\section{Examples}
\subsection{Normal priors examples}
\begin{example}
  Let $Q_n(\bx^n)$ be an i.i.d. probability distribution function, namely, $Q_n(\bx^n) = \prod_{i=1}^n Q(\bx_i)$. Obviously, $$\min_{\bx : Q_n(\bx^n) \neq 0}Q_n(\bx^n) = \BRA{\min_{\bx:Q(\bx)\neq0}Q(\bx)}^n$$ and $$-\log\BRA{\min_{\bx:Q_n(\bx^n)\neq0}Q_n(\bx^n)} = n \cdot \log\BRA{\min_{\bx:Q(\bx)\neq0}Q(\bx)},$$ which is subexponential (actually, linear in $n$).
\end{example}
\begin{example}
  Let $Q_n(\bx)$ be uniform over some set $B_n \subset \cX^n$. If $|\cX|$ is finite, then $|B_n|\leq |\cX|^n$ and $-\log\BRA{|B_n|^{-1}} \leq -\log\BRA{|\cX|^{-n}} = n\cdot\log\BRA{|\cX|}$, which is subexponential. Hence, $Q_n(\bx)$ is normal.
\end{example}

\subsection{DMC and Constant type metrics}

It is instructive to validate the fact that the family of metrics for DMCs, and more generally, metrics that are constant on types admits universal decoding over any normal prior. To see this notice that at first glance the metric family $\Theta_n$ is infinite, but since only the order on the input space induced by the metric matters, there are only a finite number of metrics that the minimum in \eqref{def:U} is achieved on. Moreover, the minimum in \eqref{def:U} occurs on some specific type and there is a polynomial number of types, so the set of metrics that receives the minimum of each type dominates the minimization and we can merge only a polynomial number of metrics to get the universal metric, which then follows to be optimal.

\subsection{Finite state metrics} \label{Example:FSC}

In this section we define a family of metrics that can be calculated using a finite state machine. This family was used in \cite{ziv1985universal}, where it was proved that the code length of the conditional Lempel-Ziv algorithm can be used as a universal decoding metric for finite state channels. For each $n$, let $\cS_n$ be the state space with $|\cS_n|$ states. A state machine is defined by the next state function $g: \cX \times \cY \times \cS_n: \rightarrow \cS_n$, and $q: \cX \times \cY \times \cS_n: \rightarrow \mathbb{R}$ the output function. The first state is $s_0$. The metric $m(\bx^n,\by^n)$ is computed as follows: Let $s_i=g(s_{i-1},x_i,y_i)$ for $i=1...n$, Then: $m(\bx^n,\by^n)=\sum_{i=1}^n q(s_i,x_i,y_i)$.

The number of next state functions is: $|\cS_n|^{|\cS_n|\cdot|\cX|\cdot|\cY|}$. Once the next state is given, only the type ($\bx,\by$ and the state) determines the metric value. In particular, the number of metrics that minimizes the pairwise error in \eqref{def:U} can be bounded by $B_n = |\cS_n|^{|\cS_n|\cdot|\cX|\cdot|\cY|}\cdot \BRA{n+1}^{|\cX|\cdot|\cY|\cdot|\cS_n|}$. Let $|\cS_n|=n^{\alpha}$. It is easy to see that $\limtoinf \frac{1}{n} \log \BRA{B_n} = 0$ for $\alpha < 1$ and $\limtoinf \frac{1}{n} \log \BRA{B_n} > 0 $ for $\alpha \geq 1$, so when $\alpha < 1$ this family admits an optimal universal decoder.

When $|\cS_n|=n$, \ie $\alpha=1$, there exists a next state function such that the state in every different time is different, and in particular, given the sent and received words, we can construct a metric that gives maximal value to these two words and smaller metric value to other codewords. This means that the minimum in \eqref{def:U} is always $Q(\bx)$ and our decoder choose the codeword according to the a priori information, independently of $\by$. Clearly, this decoder is useless and not optimal.

Notice that metrics of different ``next state function'' admits no ties in general. Thus the type based universal decoder of \cite{merhav:universal} is not optimal, while our decoder, described above, is optimal for this family of finite state metrics with $\alpha < 1$.

\ifISITpaperFull

\subsection{$k^{th}$ order Markov metrics}

A $k^{th}$ order Markov metric is a special case of a finite state metric where the state contains the last $k$ samples of $\bx,\by$. This class of metrics is interesting as it can be used to approximate several other interesting metrics, \eg, metrics of stationary ergodic channels and finite state channels with hidden channel state. It is interesting to examine the possible Markov order (as a function of $n$) that allows universal decoding. Denote the order by $k_n$. The number of states is then $|\cS_n|=\BRA{|\cX|\cdot|\cY|}^{k_n}$. Universal decoding is possible as long as $|\cS_n|=n^{\alpha}, \alpha<1$. Taking the $\log$ of both sides, we have:
$k_n < \frac{\log(n)}{\log\BRA{|\cX|\cdot|\cY|}}$.

\fi

\section{Summary and conclusions}

A general universal decoder for a family of decoding metrics is described. This decoder generalizes several known universal decoders. The decoder, termed Generalized Minimum Error Test, uses the minimum error principle. We would like to explore the usability of the minimum error principle in other areas, including universal source coding with side information. Some further research regarding the infinite (abstract) alphabet is needed. Also, extension of the results here to network cases and to channels with feedback is also investigated in on-going research.

\section{Acknowledgment}
Interesting discussions with Neri Merhav are acknowledged with thanks. Comments by the anonymous referees are greatly appreciated.

\appendices
\section{ Rate functions} \label{App:rate_func}

In order to analyze the performance of a metric decoder, we would need somehow to ``normalize'' the metric in order to remove its redundancy, \ie, the performance depends only on the order induced by the metric and not the specific value of the metric. However, for rate function (which we will define here) the value of the metric provides a hint to the performance of decoding with this metric. Throughout we will refer to a distribution $Q_n(\bx)$ as the prior and we'll assume it to be fixed. Through the whole treatment above, the output sequence $\by$ always held fixed.

\begin{definition} \mynewline

\begin{enumerate}
  \item
  A function $R:\cX_n \rightarrow \mathbb{R} \cup \BRAs{-\infty}$ is called \textbf{rate function} over the prior $Q_n(\bx)$ if:
\begin{align}\label{def:rate_func}
\PR{R(\bX) \geq t} = Q_n(R(\bX) \geq t) \leq 2^{-nt}.
\end{align}
  \item
The function $R$ is an \textbf{asymptotic rate function} if there exists a sequence $\lambda_n$ such that $\lambda_n \toinf 0$ and
\begin{align}\label{def:asym_rate_func}
Q_n(R(\bX) \geq t) \leq 2^{-n(t-\lambda_n)}.
\end{align}
  \item
Given any function $R$ (not necessarily a rate function) define:
\begin{align}\label{def:optimal_rate_func}
\Omega_R(\bx) = -\frac{1}{n}\log\BRA{Q_n(R(\bX) \geq R(\bx))}.
\end{align}
$\Omega_R(\bx)$ is the \textbf{canonical rate function} associated with the order $R$.
\end{enumerate}
\end{definition}

\begin{proposition} \label{prop:rate_func_properties} \mynewline
(a) The function $\Omega_R(\bx)$ preserves the order induced by $R$ on $\bx$. That is, $R(\bx_1) \leq R(\bx_2)$ implies $\Omega_R(\bx_1) \leq \Omega_R(\bx_2)$. If $Q_n(\bx_2) > 0$, then $R(\bx_1) < R(\bx_2)$ implies $\Omega_R(\bx_1) < \Omega_R(\bx_2)$. \\
(b) $\Omega_R(\bx)$ is a rate function. \\
(c) For any rate function $R$ and any $\bx$: $R(\bx) \leq \Omega_R(\bx)$.
\end{proposition}

\ifISITpaperFull
\begin{proof} \mynewline
(a) Let $f(t) = Q_n(R(\bX) \geq t)$. $f(t)$ is a decreasing function. $\Omega_R(\bx) = -\frac{1}{n}\log(f(R(\bx)))$ and (a) follows. \\
(b) Notice that since $\Omega_R(\bx)$ preserves the order in $R$, then $Q_n(\Omega_R(\bX) \geq \Omega_R(\bx)) = Q_n(R(\bX) \geq R(\bx)) = 2^{-n\cdot \Omega_R(\bx)}$ and the condition \eqref{def:rate_func} is met with equality. For other $t$-s that are not equal to any $\Omega_R(\bx)$ the condition can be easily verified by noting that $Q_n(\Omega_R(\bX) \geq t)=2^{-n\cdot \Omega_R(\bx)}$ for some $\bx$ such that $t < \Omega_R(\bx)$. \\
(c) If $R$ is a rate function then
\begin{align*}
  2^{-n \cdot \Omega_R(\bx)} &= Q_n(\Omega_R(\bX) \geq \Omega_R(\bx)) \\
                             &= Q_n(R(\bX) \geq R(\bx))  \\
                             &\leq 2^{-n \cdot R(\bx)}.
\end{align*}
\end{proof}
\fi

\begin{example}
Let $R(\bx) = \frac{1}{n} \log \BRA{\frac{P_n(\bx)}{Q_n(\bx)}}$ where $P_n(\bx)$ is a probability distribution, then $R$ is a rate function. To see this, notice that by using Markov's inequality
\begin{align*}
  Q_n(R(\bX) \geq t) &= Q_n(2^{n \cdot R(\bX)} \geq 2^{n \cdot t}) \\
  & \leq \EQn{2^{n \cdot R(\bX)}} \cdot 2^{-n \cdot t} \\
  & = \EQn{\frac{P_n(\bx)}{Q_n(\bx)}} \cdot 2^{-n \cdot t} \\
  & = \sum_{\bx} Q_n(\bx) \cdot \frac{P_n(\bx)}{Q_n(\bx)} \cdot 2^{-n \cdot t} = 2^{-n \cdot t}.
\end{align*}
\end{example}

\begin{definition}
We say that the function $R:\cX_n \rightarrow \mathbb{R} \cup \BRAs{-\infty}$ is an \textbf{asymptotically tight rate function} over the prior $Q_n(\bx)$ if there exists a sequence $\lambda_n$ such that $\lambda_n \toinf 0$ and
\begin{align}\label{def:tight_rate_func}
2^{-n(R(\bx)+\lambda_n)} \leq Q_n(R(\bX) \geq R(\bx)) \leq 2^{-n(R(\bx)-\lambda_n)}.
\end{align}
\end{definition}

Notice that for any function $R$, $\Omega_R(\bx)$ is a trivially asymptotically tight rate function (even not asymptotically). The following proposition gives sufficient conditions for any function $R$ to be an asymptotically tight rate function.

\begin{proposition} \label{prop:asymptotic_rate_func_cond}
For any function $R$ we have:
\begin{equation}
2^{-n \cdot \Omega_R(\bx)} \leq \EQn{ 2^{n \cdot R(\bx)} } \cdot 2^{-n \cdot R(\bx)}.
\end{equation}
In particular, if
\begin{equation} \label{eq:e1231}
2^{-n \cdot (R(\bx)+\lambda_n)} \leq 2^{-n \cdot \Omega_R(\bx)}
\end{equation}
and
\begin{equation} \label{eq:e1232}
\EQn{ 2^{n \cdot R(\bx)} } \leq 2^{n \cdot \chi_n}
\end{equation}
with $\lambda_n, \chi_n \toinf 0$, then $R$ is an asymptotically tight rate function and $\EQn{2^{n \cdot \Omega_R(\bX)}}$ is also sub-exponential.
\end{proposition}

\ifISITpaperFull
\begin{proof}
By Markov's inequality:
\begin{align*}
2^{-n \cdot \Omega_R(\bx)} =& Q_n(R(\bX) \geq R(\bx)) \\
                           =& Q_n(2^{n \cdot R(\bX)} \geq 2^{n \cdot R(\bx)}) \\
                        \leq& \EQn{2^{n \cdot R(\bX)}} \cdot 2^{-n \cdot R(\bx)}.
\end{align*}
Combining this with \eqref{eq:e1231}, \eqref{eq:e1232}: $2^{-n \cdot (R(\bx)+\lambda_n)} \leq 2^{-n \cdot \Omega_R(\bx)} \leq 2^{-n \cdot (R(\bx)-\chi_n)}$. In other words, $R$ is asymptotically tight.
Observing that \eqref{eq:e1231} implies that $2^{n \cdot (R(\bx)+\lambda_n)} \geq 2^{n \cdot \Omega_R(\bx)}$, and taking $\EQn{\cdot}$ of that we get $\E{2^{n \cdot \Omega_R(\bx)}} \leq 2^{n \cdot (\chi_n +\lambda_n)}$.
\end{proof}
\fi

Recall from definition \ref{def:normal} that the prior $Q_n(\bx)$ is \textbf{normal} if $\EQn{ 2^{n \cdot \Omega_R(\bx)} }$ is subexponential for all $R(\bx)$. The following proposition bounds $\EQn{ 2^{n \cdot R(\bx)} }$ for any rate function $R(\bx)$ and provides a condition on the prior $Q_n(\bx)$ being normal.

\begin{proposition} \label{prop:normal_prior} \mynewline
(a) Suppose that there exists some $B_n$ that bounds the rate function $R$ from above, \ie, $\PR{nR(\bX) \geq B_n}=0$. Then \begin{equation}
\EQn{2^{nR(\bX)}} \leq 1+\ln(2)\cdot B_n.
\end{equation}
(b) If $\cX_n$ is discrete, $\chi_n = -\log \BRA{ \min_{\bx : Q_n(\bx) \neq 0}Q_n(\bx) }$ is a bound on any rate function. In particular, if $\chi_n$ is subexponential then the prior $Q_n(\bx)$ is normal in the sense of definition \ref{def:normal}.
\end{proposition}

\ifISITpaperFull

\begin{proof}
(a) By \cite[Lemma 5.5.1]{koga2002information}, for a non-negative R.V. $Z$ it holds that $\E{Z} = \int_0^{\infty} \PR{Z \geq z}dz$.
\begin{align*}
  \E{2^{nR(\bX)}} &= \int_0^{\infty} \PR{2^{nR(\bX)} \geq t}dt \\
                  &= \int_0^{\infty} \PR{nR(\bX) \geq \log(t)}dt \\
                  &\overset{(1)}{\leq} 1+ \int_1^{2^{B_n}} \PR{nR(\bX) \geq \log(t)}dt \\
                  &= 1+ \int_1^{2^{B_n}} \PR{R(\bX) \geq \frac{\log(t)}{n}}dt \\
                  &\overset{(2)}{\leq} 1+ \int_1^{2^{B_n}} 2^{-n \cdot \frac{\log(t)}{n}} dt \\
                  &= 1+ \int_1^{2^{B_n}} \frac{dt}{t}\\
                  &= 1+ \ln(t)|_{1}^{2^{B_n}} = 1+ \ln(2)\cdot B_n.
\end{align*}
Inequality (1) follows because $\PR{\cdot} \leq 1$ and $\PR{nR(\bX) \geq \log(2^{B_n})}=0$ and (2) because $R$ is a rate function.

(b) For any $\bx$ with $Q_n(\bx) > 0$ we have:
\begin{align*}
n\cdot R(\bx) &\leq n\cdot \Omega_R(\bx) \\
              &= -\log \BRA{Q_n(R(\bX) \geq R(\bx))} \\
              &\leq -\log \BRA{Q_n(\bx)} \\
              &\leq -\log \BRA{ \min_{\bx : Q_n(\bx) \neq 0}Q_n(\bx)}.
\end{align*}
\end{proof}
\fi

\bibliographystyle{IEEEtran}
\bibliography{bib}

\end{document}

%% file: common_defs.tex
\usepackage{times}
\usepackage{hyperref} 
\usepackage{amsmath}
\usepackage{amssymb}
\usepackage{amsthm}
\usepackage{bbm}
\usepackage{dsfont}

\newtheorem{theorem}{Theorem}
\newtheorem{lemma}{Lemma}
\newtheorem{proposition}[theorem]{Proposition}

\newtheorem{example}{Example}
\theoremstyle{definition}
\newtheorem{definition}{Definition}
\theoremstyle{remark}
\newtheorem{remark}{Remark}

\newcommand{\mynewline}{\mbox{}\\}
\newcommand{\E}[1]{\mathbb{E} \left( #1 \right)}
\newcommand{\EQn}[1]{\mathbb{E}_{Q_n(\bx)} \left( #1 \right)}
\newcommand{\BRA}[1]{\left( #1 \right)}

\newcommand{\BRAs}[1]{\left\{ #1 \right \}}
\newcommand{\PR}[1]{Pr\left\{ #1 \right\}}

\newcommand{\bX}{{\textbf X}}

\newcommand{\cX}{{\mathcal X}}
\newcommand{\cY}{{\mathcal Y}}
\newcommand{\cS}{{\mathcal S}}
\newcommand{\cC}{{\mathcal C}}
\newcommand{\cD}{{\mathcal D}}
\newcommand{\cM}{{\mathcal M}}

\newcommand{\bx}{\textbf{x}}
\newcommand{\by}{\textbf{y}}

\newcommand{\calX}{{$\mathcal X$}}

\newcommand{\ie}{{\emph{i.e.}}}
\newcommand{\eg}{{\emph{e.g.}}}
\newcommand{\cf}{{\emph{c.f.}}}

\newcommand{\MAX}[1]{ \smash{\displaystyle\max_{#1}} }
\newcommand{\MIN}[1]{ \smash{\displaystyle\min_{#1}} }

\newcommand{\metric}[1]{#1 \left( \textbf{x}, \textbf{y} \right)}

\newcommand{\pe}[1]{p_{e,#1,\bx,\by}}

\newcommand{\pem}{\pe{m}}

\newcommand{\limtoinf}{\lim_{n \to \infty }}

\newcommand{\toinf} {\underset{n \rightarrow \infty}{\rightarrow}}
